\documentclass[acmsmall,screen,nonacm]{acmart}
\usepackage[utf8]{inputenc}
\usepackage[T1]{fontenc}
\usepackage{anyfontsize}

\usepackage{amsmath, amsfonts, amsthm, mathtools}
\usepackage{dsfont}
\usepackage{mathrsfs}
\usepackage{stmaryrd}

\usepackage[ruled,linesnumbered,noend]{algorithm2e}

\usepackage{booktabs}
\usepackage{multirow}
\usepackage{tabularx}
\usepackage{graphicx}
\usepackage[export]{adjustbox}
\usepackage{float}

\usepackage{microtype}
\usepackage{xcolor}
\usepackage{nicefrac}
\usepackage{enumitem}
\usepackage{xspace}
\usepackage{stackengine}

\usepackage[framemethod=TikZ]{mdframed}
\usepackage[most]{tcolorbox}

\usepackage{url}
\usepackage{hyperref}
\usepackage[capitalize,noabbrev]{cleveref}

\usepackage{natbib}

\input{macro}

\setcopyright{acmlicensed}
\copyrightyear{2018}
\acmYear{2018}
\acmDOI{XXXXXXX.XXXXXXX}
\acmConference[Conference acronym 'XX]{Make sure to enter the correct
  conference title from your rights confirmation email}{June 03--05,
  2018}{Woodstock, NY}
\acmISBN{978-1-4503-XXXX-X/2018/06}

\begin{document}

\title{Fast and Optimal Differentially Private Frequent-Substring Mining}


\author{Peaker Guo}
\affiliation{%
  \institution{Institute of Science Tokyo}
  \country{Japan}}
\email{peakerguo@gmail.com}

\author{Rayne Holland}
\affiliation{%
  \institution{CSIRO's Data61}
  \country{Australia}
}
\email{rayne.holland@data61.csiro.au}

\author{Hao Wu}
\affiliation{%
 \institution{University of Waterloo}
 \country{Canada}
}
\email{hao.wu1@uwaterloo.ca}






\begin{abstract}

Given a dataset of $n$ user-contributed strings, each of length at most $\length$, a key problem is how to identify all frequent substrings while preserving each user’s privacy.
Recent work by Bernardini et al. (PODS'25) introduced a $\varepsilon$-differentially private algorithm achieving near-optimal error, but at the prohibitive cost of $O(n^2\length^4)$ space and processing time.
In this work, we present a new $\varepsilon$-differentially private algorithm that retains the same near-optimal error guarantees while reducing space complexity to $O(n \length + |\Sigma| )$ and time complexity to $O(n  \length  \log |\Sigma| + |\Sigma| )$, for input alphabet $\Sigma$. 
Our approach builds on a top-down exploration of candidate substrings but introduces two new innovations: (i) a refined candidate-generation strategy that leverages the structural properties of frequent prefixes and suffixes, and (ii) pruning of the search space guided by frequency relations.  
These techniques eliminate the quadratic blow-ups inherent in prior work, enabling scalable frequent substring mining under differential privacy.

\end{abstract}





\maketitle
\section{Introduction}
\label{sec:introduction}

Modern data-driven systems operate at corpus scale.
Language models trained on user-generated text learn frequent, well-formed expressions and suppress low probability combinations.
Many natural language processing pipelines first mine frequent substrings from private corpora, then use them for next-word prediction, auto-completion, or response suggestion \citep{DBLP:conf/kdd/ChenLBCZLTWDCSW19, DBLP:conf/kdd/KannanKRKTMCLGY16, DBLP:conf/naacl/DebBS19, DBLP:conf/nips/HuLLC14}.
Despite practical success, this practice raises substantial privacy risks: a single contributor can inject rare phrases revealing medical conditions, locations or other personal details that models may memorize and later surface.

Similar concerns arise beyond natural language, for example, public transit sequences~\citep{DBLP:conf/kdd/ChenFDS12} and genomic strings~\citep{DBLP:conf/securecomm/KhatriDH19}, where frequent substrings capture, respectively, routes or motifs.
Because these corpora are linkable to individuals, naive substring mining can expose sensitive attributes.
Thus, we seek to discover global patterns for planning and analysis without revealing whether any particular user follows a specific route or carries a particular genetic variant.

\emph{Differential privacy} offers a principled mechanism for controlling this risk.
It guarantees that substrings appearing sufficiently often across the corpus can be learned and released, while the presence or absence of any single user has only a negligible effect on the output.
This naturally leads to the problem of \emph{differentially private frequent substring mining}, where the input is a database of sensitive documents and the objective is to release a set of frequent substrings while preserving strong privacy for each user.

Due to the importance of the problem, it has received a long line of research~\citep{DBLP:conf/cikm/BonomiX13,DBLP:conf/sigmod/ZhangXX16,DBLP:conf/kdd/ChenFDS12,app12042131,DBLP:conf/securecomm/KhatriDH19,DBLP:conf/nips/KimGKY21,DBLP:conf/ccs/ChenAC12,bernardini2025differentially}.
In \citeyear{bernardini2025differentially}, \citet{bernardini2025differentially} provided the first systematic theoretical study of the problem.
They presented an algorithm with rigorous additive-error guarantees and proved a matching lower bound, showing that their result is asymptotically optimal up to logarithmic factors.
A key innovation in their work is a top-down approach that identifies candidate frequent substrings over geometrically increasing substring lengths. 
This reduces both the sensitivity of the problem on similar datasets and the privacy cost incurred due to composition.

However, while theoretically significant, their approach is computationally expensive.
On a dataset of $n$ strings of length $\length$, their method requires ${O(n^{2} \length^{4})}$ processing time and space, which is infeasible at realistic scales.
For example, the Reddit dataset\footnote{A standard benchmark in frequent-substring mining.} used by Kim et al.~\citep{DBLP:conf/nips/KimGKY21} has $n \approx 10^6$ users each contributing strings with a combined length $\length \ge 3000$.
This leads to the following research question:
\begin{quote}
    \textit{\underline{Research Question}: Can we reduce the processing time and space to near-linear, while retaining asymptotically optimal error guarantees?}
\end{quote}

\subsubsection*{\bf Our Contributions.}
We answer this question in the affirmative, as summarized below.

\begin{theorem}[Informal Version of Theorem~\ref{thm:main}]\label{thm:informal-main}
    Let $\DataSet$ be a dataset of $n$ user-contributed strings over an alphabet $\Sigma$, each of length at most $\length \in \N_+$.  
    Let $\eps > 0$ and $\beta \in (0,1)$.  
    There exists an $\eps$-differentially private algorithm that, with probability at least $1 - \beta$, outputs all substrings whose frequencies in $\DataSet$ are at least $\freqUpperBound$ for some $\freqUpperBound \in \tilde{O}\paren{\tfrac{\length}{\eps}}$, 
    while using $O(n  \length  \log \card{\Sigma} + |\Sigma|)$ time and $O(n  \length + |\Sigma|)$ space.
\end{theorem}

A comparison with prior work \cite{bernardini2025differentially} is provided in Table~\ref{tab:comparison}.
There exists a trade-off between our approaches.
Our solution incurs a higher additive error in the frequency estimates, which translates to  a higher threshold for frequent substrings. 
The main difference is that we process the strings according to a binary alphabet, converting from the input alphabet $\Sigma$ if necessary.
This leads to ``longer'' strings of length $\lengthprime \doteq \length (\lceil \log |\Sigma| \rceil+1)$ and, consequently, a $ \log |\Sigma|$ factor emerges in the additive error.

This dependence on the alphabet size $|\Sigma|$ is not a practical limitation. 
In many applications, the alphabet is inherently small: for genomic data, for example, we typically have $\Sigma = \{ \texttt{A}, \texttt{C}, \texttt{G}, \texttt{T} \}$. 
Consequently, logarithmic factors in $|\Sigma|$ are negligible compared with the dependence on $n$ and $\length$, 
and do not affect the scalability of our algorithm in many real-world settings.

Despite the additional factors in the error bounds, both are $\tilde{O}(\length/\varepsilon)$ and optimal up to these polylogarithmic factors.
In addition, our reduction in performance costs is substantial and shifts the problem to a realistic scale. 
For example, our performance costs are dominated by the construction cost of a sparse suffix tree of the binary alphabet ``converted'' database, a cost admissible in practice.
Remarkably, assuming the alphabet size is $O(n\length)$, our result requires memory proportional to the size of the dataset. 

Additional discussion of related work that does not provide a theoretical treatment of the problem is relegated to Appendix~\ref{sec:related-work}.

\begin{table}[t]
    \centering
    \begin{tabular}{c|c|c|c}
        \specialrule{0.9pt}{0pt}{0pt} 
        Reference & Frequency thresholds & Time & Space \\
        \specialrule{0.8pt}{0pt}{0pt} 
         \citet{bernardini2025differentially} 
         & $
         \Theta \left( 
         \frac{\length \ln \length}{\eps}
         \left( \ln^2 \frac{n\length}{\beta} + \ln \card{\Sigma} \right)
         \right)
         $ 
         & $O(n^2 \length^4 + |\Sigma|)$ & $O(n^2 \length^4)$\\
         Theorem~\ref{thm:main} 
         & 
         $\Theta \left( 
         \frac{\length \ln^2 \lengthprime}{\eps}  \left(\ln^2 \frac{n\lengthprime }{\beta} 
         \right) \right)$
         & $O(n \lengthprime + |\Sigma|)$  & $O(n \length + |\Sigma|)$ \\
         \specialrule{0.9pt}{0pt}{0pt} 
    \end{tabular}
    \caption{
    Comparison between our result and prior work by \citet{bernardini2025differentially}.
    $\lengthprime = \length \cdot (\lceil \log |\Sigma|\rceil+1) $.
    }
    \vspace{-5mm}
    \label{tab:comparison}
\end{table}

\subsubsection*{\bf Algorithm Overview}

The main bottleneck from prior work arises from \emph{how} candidate substrings are generated.
For example, given the set $\Candidates{k}$ of frequent substrings of length $k$, \citet{bernardini2025differentially} test the frequency of {every} concatenated pair $s_1 \cdot s_2$, for $s_1,s_2 \in \Candidates{k}$, to determine frequent substrings of length $2k$.
This leads to $|\Candidates{k}|^2$ candidates per phase and induces a quadratic cost in both runtime and memory.

To mitigate the combinatorial blow-up, we introduce a novel top-down algorithm that strategically prunes the search space.
Our starting point is the observation that, conditioned on having correctly identified the set $\Candidates{k}$ of frequent substrings of length $k$, 
every longer frequent substring of length $k+t$, for $t \in \{1,\ldots, k\}$, \emph{must extend some element of $\Candidates{k}$ along a suffix that appears in $\Candidates{k}$}.

Leveraging this structure, we build a single compact trie $T_k$ from the suffixes of $\Candidates{k}$.
Subsequently, we explore candidates that arise from traversing the concatenated trees $s \circ T_k$, for $s \in \Candidates{k}$, where $s$ is viewed as a single-path tree whose last node is connected to the root of $T_k$.
During this top-down exploration, we {prune} entire subtrees with (noisy) frequency estimates below a predefined threshold.
As truly frequent substrings must follow from existing frequent prefixes, this pruning removes large portions of the search space without impeding correctness.
As we reuse $T_k$ over each traversal, the result is that the time spent in each phase is proportional to the combined length of the frequent substrings that belong to the phase.
Thus, the total search cost is proportional to the combined length of all frequent substrings.

The full algorithm is $\varepsilon$-differentially private and relies on noisy frequency estimates with additive-error $\Tilde{O}(\tfrac{\length}{\eps})$, which is optimal up to polylogarithmic factors.

\subsubsection*{\bf Organization}
Section~\ref{sec:problem-setup} introduces the necessary definitions and formalizes the differentially private frequent substring mining problem.
Section~\ref{sec:preliminary} reviews the technical preliminaries on differential privacy and the string data structures used by our algorithm.
In Sections~\ref{sec:algorithm} and~\ref{sec:analysis}, we present our algorithm together with an analysis of its privacy, utility, and asymptotic cost. 
Finally, Section~\ref{sec:conclusion} offers some concluding remarks.

\section{Problem Setup}\label{sec:problem-setup}

In this section, we formally define the problem of differentially private frequent substring mining.

\subsection{Notation}
We use $\N$, $\Z$, and $\R$ to denote the sets of \emph{natural numbers}, \emph{integers}, and \emph{real numbers}, respectively, and write $\N_+$ for the positive naturals, $\R_{\ge 0}$ for the nonnegative reals, and $\R_+$ for the positive reals.
For each $n \in \N$, 
we define 
$[n] \doteq \set{1, \dots, n }$.

Let $\DataSet \doteq \{\Document{1}, \ldots, \Document{n}\}$ 
be a dataset of $n$ strings drawn from $\Sigma$, 
where each string has length at most $\length \in \N_+$ 
(i.e., $\Document{i} \in \Sigma^{1:\length}$). 
For each string $P \in \Sigma^{1:\length}$, we define its frequency in the dataset $\DataSet$ as
\(
    \freq{\DataSet}{P}
    \;\doteq\;
    \sum_{i=1}^n \freq{\Document{i}}{P},
\)
where $\freq{\Document{i}}{P}$ denotes the number of occurrences of $P$ in $\Document{i}$.
We view the frequencies of all strings in $\Sigma^{1:\length}$ as the vector
\(
    \freq{\DataSet}
    \;\doteq\;
    \bigl( \freq{\DataSet}{P} \bigr)_{P \in \Sigma^{1:\length}}.
\)

\subsection{Privacy Guarantee}

In our setting, the dataset $\DataSet$ is contributed by $n$ users, where the string $\Document{i}$ belongs to the $i$-th user.  
To protect user privacy, the released frequent strings must be insensitive to the content of any individual string, so that modifying one user's string does not significantly affect the output. 

This requirement is formalized through the standard notion of differential privacy (DP). 
We first recall the general definition of DP for arbitrary randomized algorithms, and later instantiate it for the problem considered in this work.

\begin{definition}[$\eps$-Indistinguishability]
    \label{def:indistinguishability}
    Let $\eps \in \R_{\ge 0}$, and let $(\cZ, \cF)$ be a measurable space. 
    Two probability measures $\mu$ and $\nu$ on this space are 
    \emph{$\eps$-indistinguishable} if
    \begin{equation}
        e^{-\eps} \cdot \mu(E) 
        \;\le\;
        \nu(E)
        \;\le\;
        e^{\eps} \cdot \mu(E),
        \qquad
        \forall E \in \cF.
    \end{equation}
    Likewise, two random variables are $\eps$-indistinguishable if their induced distributions are.
\end{definition}

\begin{definition}[$\eps$-Differential Privacy~\citep{DworkMNS06}]
    \label{def:DifferentialPrivacy}
    Let $\cM : \cY \to \cZ$ be a randomized algorithm and let $\sim$ be a symmetric neighboring relation on $\cY$.
    For parameters $\eps \in \R_{\ge 0}$,  
    we say that $\cM$ is \emph{$\eps$-differentially private} if for every neighboring pair $y \sim y'$,  
    the distributions of $\cM(y)$ and $\cM(y')$ are $\eps$-indistinguishable.
\end{definition}

\subsection{Differentially Private Frequent Substring Mining.}

We begin by formally defining what it means for a substring to be considered frequent.

\begin{definition}
    \label{def:threshold}
    Given a threshold $\freqUpperBound \in \mathbb{R}_+$, a string $P \in \Sigma^{1:\length}$ is called \emph{$\freqUpperBound$-frequent} if $\freq{\DataSet}{P} \ge \freqUpperBound$.  
    Similarly, given a threshold $\freqLowerBound \in \mathbb{R}_+$, we say that $P$ is \emph{$\freqLowerBound$-infrequent} if $\freq{\DataSet}{P} \le \freqLowerBound$.
\end{definition}

\begin{definition}[Inclusion-Exclusion Criterion]
\label{def:candidate-property}
Given a dataset $\DataSet$, a set of strings $\Candidates \subseteq \Sigma^{1:\length}$ satisfies the \emph{Inclusion-Exclusion Criterion} for thresholds $\freqUpperBound \ge \freqLowerBound \ge 0$ if:
\begin{enumerate}[leftmargin=7mm]
    \item[(i)] Every $\freqUpperBound$-frequent string $P \in \Sigma^{1:\length}$ is included in $\Candidates$.
    \item[(ii)] Every $\freqLowerBound$-infrequent string $P \in \Sigma^{1:\length}$ is excluded from $\Candidates$.
\end{enumerate}
We call such a set $\Candidates$ a \emph{$(\freqUpperBound, \freqLowerBound)$-IE set}.
\end{definition}

The \emph{differentially private frequent substring mining problem} aims to identify a \emph{$(\freqUpperBound, \freqLowerBound)$-IE set} $\Candidates$ while simultaneously:
\begin{enumerate}[leftmargin=7mm]
    \item minimizing the processing time and space;
    \item minimizing $\freqUpperBound$ while maximizing $\freqLowerBound$; and
    \item ensuring $\varepsilon$-differential privacy across the execution of the algorithm.
\end{enumerate}

Sometimes, we also want to output, for each $P \in \Candidates$, an approximation $\noisyfreq{\DataSet}{P}$ of its true frequency $\freq{\DataSet}{P}$.  
By setting $\noisyfreq{\DataSet}{P} = 0$ for all $P \notin \Candidates$, the vector $\noisyfreq{\DataSet}$ becomes a sparse approximation of the full frequency vector $\freq{\DataSet}$, with at most
\(
    O\!\left(\sum_{i \in [n]} \card{\Document{i}}^2 / \freqLowerBound\right)
    \;\subseteq\;
    O(n \length^2 / \freqLowerBound)
\)
non-zero entries.
Consequently, the problem is closely connected to the classical task of releasing a differentially private sparse histogram.

Further, \citet{bernardini2025differentially} present an algorithm achieving
\(
    \freqUpperBound, \freqLowerBound 
    \!\in\!
    {\Theta} \bigparen{
        \frac{\length \ln \length }{\eps}
        \bigl(\ln^2 \frac{n\length}{\beta} + \ln \card{\Sigma}\bigr)
    },
\)
with probability at least $1 - \beta$ for any given $\beta \in \paren{0, 1}$, and prove a lower bound showing that
\(
    \freqUpperBound - \freqLowerBound
    \in
    \Omega\!\left( \min\!\left\{n, \,\eps^{-1}\length \ln\card{\Sigma}\right\} \right).
\)
Since $\freqLowerBound \ge 0$, this implies
\(
    \freqUpperBound
    \in
    \Omega\!\left( \min\!\left\{n, \,\eps^{-1}\length \ln\card{\Sigma}\right\} \right).
\)
Thus, both $\freqUpperBound$ and $\freqLowerBound$ are asymptotically tight up to logarithmic factors in their algorithm.
However, their approach requires $O(n^2\length^4)$ space and processing time.
Therefore, the main objective is to develop algorithms that are substantially faster and more space-efficient, while retaining the same asymptotic error guarantees.

\subsubsection*{DP for Frequent Substring Mining.}
To instantiate this definition for our problem, we interpret $\cM : \cY \rightarrow \cZ$ as the randomized algorithm that publishes the (possibly noisy) frequent substrings.  
Here, $\cY$ is the space of all possible datasets, i.e.,
\(
    \cY = (\Sigma^{1:\ell})^n,
\)
and $\cZ$ is the collection of all subsets in $\Sigma^{1:\ell}$.
We treat two datasets $\DataSet, \DataSet' \in \cY$ as \emph{neighbors} if they differ in the string of exactly one user; that is,
\[
    \DataSet \sim \DataSet'
    \quad\text{iff}\quad
    \exists\, i \in [n] \text{ such that }
    \Document{i} \neq \Document{i}' 
    \text{ and }
    \Document{j} = \Document{j}' \text{ for all } j \neq i.
\]
A frequent substring mining algorithm is $\eps$-DP if it satisfies Definition~\ref{def:DifferentialPrivacy} under this neighboring relation.

\section{Preliminaries}
\label{sec:preliminary}

In this section we introduce the technical preliminaries on differential privacy and the string data structures that our algorithm relies on.

\subsection{Privacy}\label{subsec:privacy}

\begin{lemma}[Composition~\citep{DR14}]
    \label{lem:composition}
    Let $\cM_1, \cM_2, \dots, \cM_k$ be a sequence of randomized algorithms where
    \(
        \cM_1 : \cY \to \cZ_1,
    \)
    and
    \(
        \cM_i : \cZ_1 \times \cdots \times \cZ_{i-1} \times \cY \to \cZ_i,
    \)
    for
    \(
        i = 2,\dots,k,
    \)
    i.e., $\cM_i$ may take as input the outputs of $\cM_1,\dots,\cM_{i-1}$ and the dataset in $\cY$.
    Suppose that for each $i\in[k]$ and for every fixed $a_1\in\cZ_1,\dots,a_{i-1}\in\cZ_{i-1}$,
    the randomized mapping $\cM_i(a_1,\dots,a_{i-1},\cdot)$ is $\eps_i$-differentially private.
    Then the algorithm
    \(
        \cM : \cY \to \cZ_1\times\cdots\times\cZ_k,
    \)
    obtained by running the $\cM_i$'s sequentially, is $\eps$-differentially private for
    \(
    \eps \!=\! \sum_{i=1}^k \eps_i.
    \)
\end{lemma}

\begin{lemma}[Tail Bound for Laplace Noise]
\label{lem:laplace-tail}
Let $Z \sim \mathrm{Laplace}(\sigma)$ be a Laplace random variable with scale parameter $\sigma > 0$. Then, for any $\beta \in (0,1)$,
\[
    \P{
        |Z| > \sigma \ln \tfrac{1}{\beta} 
    }\le \beta.
\]
\end{lemma}

\begin{lemma}[Laplace Mechanism~\citep{DworkMNS06}]
    \label{lem:lap_mech}
    Let $\cY$ be a domain equipped with a neighboring relation $\sim$, and let $F : \cY \to \mathbb{R}^d$ satisfy
    \[
        \| F(y) - F(y') \|_1 \le \Delta
        \quad \text{for all } y \sim y', \text{ for some } \Delta > 0.
    \]
    For each privacy parameter $\eps > 0$, the following mechanism is $\eps$-differentially private:
    \[
        \cM(y) = F(y) + (Z_1,\dots,Z_d),
        \quad\text{where } Z_i \stackrel{\text{iid}}{\sim} \mathrm{Laplace}(\Delta/\eps).
    \]
\end{lemma}

The Binary Tree Mechanism represents each prefix sum over a sequence of length $T$
by summing the noisy counts of $\lceil \log T \rceil + 1$ dyadic intervals. 
Independently, at each dyadic node we add Laplace noise. 
The properties of the Binary Tree Mechanism are summarized below.

\begin{lemma}[Binary Tree Mechanism~\citep{ChanSS11,DworkNPR10}]
    \label{lem:binary-tree}
    Fix $\eps_0 \in \R_+$ and let $\sigma \doteq \eps_0^{-1}\cdot \log d$.
    Then there exists a mechanism $\cB^{\sigma} : \R^d \to \R^d$, with parameter $\sigma$, with the following guarantees:
    \begin{itemize}[leftmargin=5mm, label=$\triangleright$]

        \item \textbf{Privacy Guarantee.}
        For all $x, x' \in \R^d$, the outputs $\cB^{\sigma}(x)$ and $\cB^{\sigma}(x')$ are
        $\eps_0 \cdot \norm{x - x'}_1$-indistinguishable.

        \item \textbf{Utility Guarantee.}
        Let $x \in \R^d$ and $\beta \in (0,1)$.
        Define the prefix-sum vector
        \(
            y \doteq 
            \bigl(
                x_1,\;
                x_1 + x_2,\;
                \ldots,\;
                {\textstyle\sum_{i \in [d]} x_i}
            \bigr).
        \)
        Then
        \[
            \P{
                \norm{\cB^{\sigma}(x) - y}_\infty
                \;>\;
                 2\sqrt{2} \cdot \sigma \cdot \ln \frac{2d}{\beta}
            }
            \;\le\;
            \beta.
        \]

        \item \textbf{Streaming Implementation.}
        The mechanism $\cB^{\sigma}$ supports online access: for each $t \in [d]$, upon receiving the first $t$ coordinates of $x$, it outputs the first $t$ coordinates of $y$.
        It uses $O(\log d)$ additional space and computes each coordinate of $y$ in $O(\log d)$ time.
    \end{itemize}
\end{lemma}

\subsection{Suffix Trees}\label{subsec:suffix-trees}

\subsubsection*{\bf Strings.}
Let $\Sigma$ be an alphabet.
A \emph{string} is an element of $\Sigma^*$.
For a string $S$,
its length is $|S|$, and $S[i]$ denotes its $i^{\text{th}}$ character.
For integers $l \leq r$, let
$\Sigma^{l:r} \doteq \{ S \in \Sigma^* \mid l \leq |S| \leq r \}$.
A  \emph{substring} of $S$, 
starting at position $i$ 
and ending at position $j$, is written as $S[i \ldots j]$. 
A substring $S[1\ldots j]$ is called a \emph{prefix} of $S$, 
while $S[i \ldots |S|]$ is called a \emph{suffix} of $S$.
Let $\suffix{S}$ denote the set of all suffixes of $S$.
For a multiset of strings $\DataSet$, let $\suffix{\DataSet} \doteq \cup_{S \in \DataSet} \suffix{S}$.
The \emph{concatenation} of strings $S$ and $T$ is $S \circ T$.
For two strings $P$ and $S$, an \emph{occurrence} of $P$ in $S$
is a position $i$ such that $S[i \ldots i+|P|-1] = P$.
The \emph{frequency} of $P$ in $S$, denoted by $\freq{S}{P}$,
is the number of occurrences of $P$ in $S$.

\subsubsection*{\bf Tries}
Given a multiset of non-empty strings $\DataSet = \{S_1, \ldots S_k\} \subseteq \Sigma^*$, 
the \emph{trie} $\trieof{\DataSet}$ over $\DataSet$ is a rooted tree in which edges are labeled with a single character from $\Sigma$; 
no two outgoing edges from any node share the same label; and
for each string $S_i$ in $\DataSet$, there exists a unique 
root-to-leaf path whose concatenation of edge labels spells $S_i$ exactly.
For any prefix $P$ of $S_i$, let $\trienode{P}$ denote the unique node whose path from the root spells $P$.

\subsubsection*{\bf Suffix trees.}
For $S \in \Sigma^*$, the \emph{suffix tree} is the \emph{path-compressed} trie of all suffixes of $S$,
obtained by contracting each maximal non-branching path into a single edge whose
label is the concatenation of the labels along that path.
It contains $O(|S|)$ nodes and can be constructed in $O(|S|)$ time,{assuming $\Sigma$ is a polynomial-sized integer alphabet}~\citep{conf/swat/1973/weiner, journal/jacm/1976/mccreight, journal/algorithmica/1995/ukkonen, DBLP:conf/focs/Farach97}.
For a node $u$ in a suffix tree,

\begin{itemize}[leftmargin=5mm,label=$\triangleright$]
\item 
$\str{u}$ denotes the \emph{path label} of $u$.
\item 
$\parent{u}$ denotes the parent of $u$.
\item 
For a symbol $c \in \Sigma$,
$\child{u}{c}$ denotes the child node $v$ of $u$ such that 
the label of the edge $uv$ begins with $c$. 
If such $v$ does not exist, then  $\child{u}{c} = \varnothing$.
\item 
$\children{u} \doteq \{ (\child{u}{c}, c) : c \in \Sigma \text{~and~}  \child{u}{c} \neq \varnothing \}$.
\end{itemize}

Every substring $P$ of $S$ can be located by traversing the suffix tree from the root following the characters in $P$.
Such a traversal takes $O(|P|)$ time and finishes at a unique \emph{locus} (either a node or a position within an edge). 
If the traversal ends inside the edge leading to $u$, the $P$ equals the prefix of $\str{u}$ of the appropriate length.
The frequency of $P$ equals the number of leaf nodes in the subtree rooted at $u$.
Since all substrings whose loci lie along the same compressed edge share the same frequency, we may store exact frequencies at branching nodes.
This requires $O(|S|)$ additional space.

We make use of the following construction for storing a subset of evenly spaced suffixes of a string. 

\begin{definition}[$r$-Spaced Sparse Suffix Tree]
\label{def:index_align_SST}
Let $S$ be a string and let $r\ge 1$.
The \emph{$r$-spaced sparse suffix tree} of $S$ is the path-compressed trie built from all the suffixes $S[j\ldots |S|]$ whose starting index $j$ satisfies $j\equiv 1 \pmod{r}$.
\end{definition}

The following result states the properties of the $r$-spaced sparse suffix tree built on the $r$-spaced suffixes of a multiset of strings.

\begin{lemma}[\cite{karkkainen1996sparse}]
\label{lem:r_spaced_SST_multi}
For $r\ge 1$ and a multiset $\{S_1,\ldots,S_n\}$, where $|S_i| \equiv 0 \pmod{r}$,
let $T = S_1\#_1^r S_2\#_2^r\cdots S_n\#_n^r$.  
Then the $r$-spaced sparse suffix tree of $T$ can be constructed in
$O(|T|)$ time and $O(|T|/r)$ space. 
\end{lemma}

The result requires that string lengths are a multiple of $r$ so that the correct ($r$-spaced) suffixes of all strings are stored in the tree.
This is also why the string delimiters are repeated $r$ times.

For $\DataSet = \{S_1,\ldots,S_n\}$,
during the construction of the $r$-spaced suffix tree of $T = S_1\#_1^r S_2\#_2^r\cdots S_n\#_n^r$, we may augment each explicit node $v$ with a frequency counter that stores the number $\freq{\DataSet}{\str{v}}$ of leaf nodes in the subtree rooted at $v$.
This can be maintained without affecting the asymptotic bounds in Lemma~\ref{lem:r_spaced_SST_multi}.

\subsection{Heavy-Light Decomposition}

The \emph{heavy-light decomposition} of a rooted tree $T$
partitions $T$ into edge-disjoint paths.
Specifically, for a non-leaf node $u$ in $T$,
let $v$ be the child of $u$ with the largest number of nodes 
in the subtree rooted at $v$;
we say edge $uv$ is a \emph{heavy} edge,
and for any other child $w$ of $u$,
we say edge $uw$ is a \emph{light} edge.
The decomposition is defined as follows.
First, the path from the root to a leaf where each edge is a heavy edge is a heavy path.
Then, for each edge $uv$ such that $u$ is on the heavy path and $v$ is not, recursively decompose the subtree rooted at $v$.
\begin{lemma}[\cite{SleatorT83}]
    \label{lem:hl-decomposition}
    If tree $T$ has $m$ nodes, then a heavy-light decomposition of $T$ can be constructed in $O(m)$ time, and each root-to-leaf path of $T$ has at most $\lceil \log m \rceil$ light edges.
\end{lemma}

Intuitively, the lemma is correct because traversing a light edge reduces the current subtree size by at least a factor of two. 
Therefore, a root-to-leaf path can encounter at most $\lceil \log m \rceil $ such reductions before reaching a leaf.

\section{Algorithm}
\label{sec:algorithm}

The main properties of our algorithm are summarized below.
\begin{theorem}
    \label{thm:main}
    Let $\DataSet = \{\Document{1}, \ldots, \Document{n}\}$ be a dataset of $n$ user-contributed strings over an alphabet $\Sigma$, each of length at most $\length \in \N_+$.  
    Let $\eps > 0$, $\beta \in (0,1)$, and $\lengthprime = \length (\lceil\log |\Sigma| \rceil+1) $, and set
    \[
    \eps_0 = \frac{ \eps/ \log \length }{4\length \log( n \lengthprime)}
    \quad \text{and} \quad    
    \tau^* = { \eps_0^{-1} \cdot  \paren{ \log \lengthprime } \cdot \ln(n\lengthprime /\beta)}.
    \]
    Then, for thresholds
    \begin{align*}
        \freqUpperBound &= 9 \cdot \tau^*  
        \quad \text{and} \quad
        \freqLowerBound \geq  \length \paren{\log \lengthprime} 
    \end{align*}
    there exists an $\eps$-differentially private algorithm that, with probability at least $1 - \beta$, outputs a set $\Candidates$ that satisfies the $(\freqUpperBound, \freqLowerBound)$-Inclusion-Exclusion Criterion, 
    while using $O(n  \lengthprime + |\Sigma|)$ time and $O(n\length + |\Sigma| )$ space.  
\end{theorem}

In the remainder of this section, we give a high-level overview of the algorithm that obtains Theorem~\ref{thm:main}.
The formal analysis is delegated to \cref{sec:analysis}.

\subsection{Prior Techniques}\label{subsec:prior-techniques}
Before diving into the details of the algorithm, we examine earlier approaches to shed light on why privately mining frequent substrings at scale remains a challenging problem.

A naive approach would be to directly construct a noisy ($\eps$-differentially private) histogram of all strings in $\Sigma^{1:\length}$.
Subsequently, we output only those strings whose noisy frequencies exceed a prescribed threshold.
However, this direction entails two limitations.
First, the frequency vector $(\freq{\DataSet}{p})_{p \in \Sigma^{1:\ell}}$ has sensitivity $\Omega(\length^2)$, thereby requiring a large amount of noise for effective privacy.
Second, as the number of strings is exponential in the string length, this approach is computationally infeasible.

To mitigate both universe explosion and problem sensitivity,~\citet*{bernardini2025differentially} adopt a top-down strategy.
Starting from a single character, they expand the search to longer candidate substrings if the noisy counts of their prefixes exceed a threshold. 
This allows the exponential search space to be significantly pruned.
However, their approach is still limited by the number of candidate substrings they consider.
For example, given the set $\Candidates{k}$ of all $\freqUpperBound$-frequent substrings of length $k$, to generate $\Candidates{2k}$ they exhaustively consider \emph{all} $|\Candidates{k}|^2$ concatenations of substrings in $\Candidates{k}$.
This leads to a (prohibitive) quadratic blow-up in runtime and space.

In addition, they only consider candidates at a logarithmic number of substring lengths and, subsequently, interpolate \textit{all possible} frequent substrings at the remaining string lengths.
This has the benefit of reducing the sensitivity of the search algorithm, thereby reducing privacy noise.
However, this interpolation step also leads to at least a quadratic blow-up in runtime and space.

\subsection{Overview}

To overcome these challenges, 
we adopt a top-down search procedure that ventures only into regions of the subspace where substrings might plausibly be frequent.
However, compared to \citet{bernardini2025differentially},  
we are more prudent in our generation of candidates, focusing effort only where it matters. 

\subsubsection{Preprocessing}

We process the strings in $\DataSet$ using the binary alphabet.
Each symbol in $\Sigma$ is converted to a binary codeword followed by the terminal $\texttt{\$}\notin \Sigma$.
For $r=\lceil\log |\Sigma|\rceil+1$, 
this increases the string length to $\lengthprime = \length r  $.
The encoding is done with the codebook $E:\Sigma \to \{ x \texttt{\$} \mid x \in \{\texttt{0,1}\}^{r-1} \}$.
Assuming each character in $\Sigma$ can be stored in $O(1)$ words, the codebook can be stored in a direct-access array in $O(|\Sigma|) $ space. 
The codebook is extended to strings through concatenation.
That is, for a string $S = x_1 \cdots x_n,$ $E(S) \doteq E(x_1) \cdots E(x_n)$.

The terminal at the end of each encoded character is necessary to prevent the extraction of frequent substrings that are not character aligned.
For example, the substring `\texttt{AC}', for $\texttt{A},\texttt{C} \in \
\Sigma$, may occur frequently in $\DataSet$.
However, if the symbols are encoded as $E(\texttt{A})=\texttt{0000}$ and $E(\texttt{C})= \texttt{0101}$, we may incorrectly output `\texttt{0001}' if the boundary between \texttt{A} and \texttt{C} is not recognized. 

Further, when operating on the binary alphabet, we need to ensure that output substrings are well-defined.
That is, we want to avoid outputting frequent substrings that cannot be correctly decoded.
Continuing the example above, the substring `$\texttt{00\$0101\$}$' can be frequent in the encoded dataset, but it does not decode to a substring in $\Sigma^*$.
To enable correct outputs, we only generate \emph{character-aligned} substrings, which are defined as follows.

\begin{definition}[Character-aligned Strings]\label{def:char_aligned_substr}
Let $b \doteq \lceil\log|\Sigma|\rceil$ and $r \doteq b+1$. 
Fix an injective block encoding
$E:\Sigma\to\{x \texttt{\$} \mid x \in \{\texttt{0},\texttt{1}\}^b\}$ and extend it to strings by concatenation.
A length-$k$ string $P\in\{\texttt{0},\texttt{1},\texttt{\$}\}^k$ is \emph{character-aligned} if there exist integers $a\ge 0$ and $0\le c\le r$ such that $k=ar+c$ and {$P[1 \ldots ar]$ can be written as $B_1 \cdots B_a$ where each $B_i \in E(\Sigma)$.}
\end{definition}
Note that only the last block of a character-aligned substring can be incomplete.
For example, for $b=2$, $r=3$, and $E$ defined in the caption of Figure~\ref{fig:freq_substring_mining},
$P=\texttt{01\$10\$1}$ is character-aligned since there exist $a=2$ and $c=1$ such that $7= 2\cdot 3 + 1$ and $P[1 \ldots 6] = \texttt{01\$} \cdot \texttt{10\$}$.
The character-aligned substrings of a dataset have the following sensitivity.
\begin{lemma}
\label{lem:sens_bound_aligned}
Let $r \doteq \lceil \log |\Sigma | \rceil +1$. 
For each $0\le k\le \lengthprime=\length r$, let $\mathcal{A}_k$ be the family of length-$k$ \emph{character-aligned encoded substrings} $P\in\{\texttt{0},\texttt{1},\texttt{\$}\}^k$. 
Then the $L_1$-sensitivity of the vector $(\freq{\DataSet}{P})_{P\in \mathcal{A}_k}$ is bounded by $2\length$.
\end{lemma}
\begin{proof}
Changing one record $\Document$ in $\DataSet$ can affect only those counts $\freq{\DataSet}{P}$ that come from substrings $P$ occurring in $\Document$ itself.
Under the binary encoding, the encoding string $E(\Document)$ has length $\length r$.
Character-aligned substrings may begin only at positions corresponding to original character boundaries, that is, at indices $1, r+1, 2r+1, \ldots, (\length-1)r -1$.
Therefore, any character-aligned string of length $k$ has at most $\lengthprime/r-k+1\le \length$ possible starting positions in $\Document$. 
Consequently, replacing $\Document$ can change at most $\length$ of the counts $\freq{\DataSet}{P}$ by $+1$, and at most $\length$ of them by $-1$, giving total $L_1$ change at most $2\length$.
\end{proof}

\subsubsection{Frequent Substring Mining}

The procedure unfolds over $\lceil \log \length \rceil $ phases and takes the threshold parameter $\freqLowerBound$ as input (see Definition~\ref{def:threshold}).  
Let $\Sigma_E  \doteq \{x \texttt{\$} \mid x \in \{\texttt{0},\texttt{1}\}^{r-1}\}$. 
In phase $i$, with $k = r\cdot 2^{i-1}$, the algorithm privately outputs the frequent-substring sets for strings in $\Sigma_E^{(k+1):2k} \equiv \Sigma^{(k+1):2k}$.
The process is inductive, extending existing frequent substrings of length $k$ to generate more candidates.
Because extensions are carried out bit by bit, the algorithm must also determine frequent substrings at intermediate lengths (that is, substrings of length $ar+c$ for $c <r$). 
This bitwise decomposition accelerates the pruning process.
If we were operating exclusively on the general alphabet $\Sigma$, each substring extension would need to consider $O(|\Sigma|)$ additional candidates.
Working with a binary alphabet, we only consider at most two additional candidates per extension, at the cost of $O(\log |\Sigma|)$ additional rounds.
Thus, probing frequencies at the bit level yields a more efficient extension procedure.

Beginning at the base phase, $i=1$, we construct the frequency vector 
$(\freq{\DataSet}{\gamma})_{\gamma \in \{\texttt{0},\texttt{1}\}^{r-1}\texttt{\$}}$.
By Lemma~\ref{lem:sens_bound_aligned}, this vector has sensitivity bounded by $2\length$.
Therefore, we can construct noisy counts $(\freqnoisy{\DataSet}{\gamma})_{\gamma \in \{\texttt{0},\texttt{1}\}^{r-1}\texttt{\$}}$ that, with high probability, satisfy $\varepsilon_0$-differential privacy with $\Tilde{O}(\ell/\varepsilon_0)$ additive error by Lemmas~\ref{lem:laplace-tail}~and~\ref{lem:lap_mech} . 
The set $\Candidates{r}$, of character-aligned substrings of length $r$, is then constructed by selecting substrings with noisy frequency above a threshold $\tau \in \Tilde{O}(\length/\varepsilon)$ to be defined later.

At each subsequent phase $i>1$, with $k=r\cdot 2^{i-1}$, we compute $\Candidates{k+1}, \ldots, \Candidates{2k}$ from the substrings in $\Candidates{k}$.
The process is built on the following key observation.

\begin{lemma}
    \label{lem:frequent_sub}
    Given a dataset $\DataSet$ over $\Sigma$, let $\Candidates{k}$ denote a set that includes all substrings of length $k$ that are $\freqUpperBound$-frequent.
    Let $\CandidateExtension{k,t}$ be the set consisting of all strings $\bar{s}$ of length $k + t$, for some $t \in [k]$, such that
    \begin{align*}
        &\text{(i)}\ \bar{s}[1 \ldots k] \in \Candidates{k}, \\   
        &\text{(ii)}\ \exists\, s' \in \Candidates{k} \text{ such that } \bar{s}[k + 1 \ldots k + t] \text{ is a suffix of } s'.
    \end{align*}
    Then for each $t \in [k]$, every $\freqUpperBound$-frequent substring of length $k + t$ is contained in $\CandidateExtension{k,t}$.
\end{lemma}

\begin{proof}
    Assume a substring $\hat{s}$, with $|\hat{s}| = k +t$, is $\freqUpperBound$-frequent.
    It suffices to show that $\hat{s} \in \CandidateExtension{k,t}$.
    By Definition~\ref{def:threshold}, $\freq{\DataSet}{\hat{s}} \geq \freqUpperBound$.
    For (i), the prefix condition, consider the prefix $p = \hat{s}[1 \ldots k]$.
    For any occurrence of substring $\hat{s} = \Document[i \ldots i + k + t - 1]$, for $\Document \in \DataSet$, $p = \Document[i \ldots i + k - 1]$ occurs in the same position.
    Thus,
    $ 
        \freq{\DataSet}{p} \geq \freq{\DataSet}{\hat{s}}.
    $ 
    By assumption, every substring with frequency at least $\freqUpperBound$ is in $\Candidates{k}$.
    Therefore, $p \in \Candidates{k}$.
    For (ii), the suffix condition, let $q = \hat{s}[k + 1 \ldots k + t]$ denote the length-$t$ suffix of $\hat{s}$ and let $s^{\prime} = \hat{s}[t + 1 \ldots k + t]$ denote the length-$k$ suffix of $\hat{s}$.
    Since $\hat{s}$ is $\freqUpperBound$-frequent, every length-$k$ substring of $\hat{s}$ is also
    $\freqUpperBound$-frequent; hence $s' \in \Candidates{k}$.
    Moreover, $q$ is a suffix of $s'$.
    Consequently, $\hat{s} \in \CandidateExtension{k,t}$, as required.
    Therefore, $\hat{s}\in \CandidateExtension{k,t}$.
\end{proof}
\noindent 
Lemma~\ref{lem:frequent_sub} is visualized in Figure~\ref{fig:prefix-suffix-freq}.

\begin{figure}[t]
    \centering
    \includegraphics[width=0.3\linewidth]{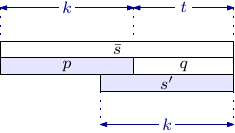}
    \caption{
        Illustration of Lemma~\ref{lem:frequent_sub}, with strings in $\Candidates{k}$ highlighted.
    }
    \label{fig:prefix-suffix-freq}
\end{figure}

\paragraph*{Candidate Extension.}
To facilitate candidate generation, we construct the tree $T_k$ of the character-aligned suffixes in $\Candidates{k}$.
As $T_k$ is computed through post-processing $\Candidates{k}$, it incurs no privacy cost.
By Lemma~\ref{lem:frequent_sub},
conditioned on $\Candidates{k}$ satisfying the Inclusion-Exclusion Criterion,
any $\freqUpperBound$-frequent character-aligned substring of length $k+t$ must:
\begin{enumerate}
    \item\label{item:freq_cond_1} begin with some $s \in \Candidates{k}$, and
    \item\label{item:freq_cond_2} continue with a substring $q$ that forms a length-$t$ path from the root of $T_k$.
\end{enumerate}
In other words, there exists $u\in T_k$, where $\str{u} = q$ and $|q|=t$.
Therefore, to enumerate all $\freqUpperBound$-frequent substrings of lengths $k+1, \ldots 2k$, it suffices to search over path concatenations of the form
\[
 s \circ \str{u}, \quad u \in T_k, \, s \in \Candidates{k}.
\] 

\begin{figure}[t]
    \centering
    \includegraphics[width=\linewidth]{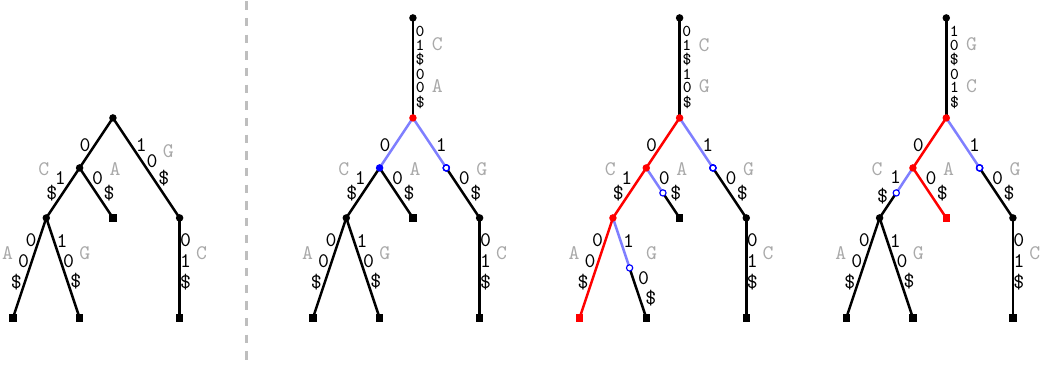}
    \caption{
    Example of a phase $i=2$ in \cref{alg:DPFS} with
    $\DataSet = \{ \texttt{CGCA}, \texttt{CGCA}, \texttt{CATA}\}$, 
    $\freqLowerBound= 2$, and
    block encoding $E$ specified by
    $\texttt{A} \mapsto \texttt{00\$}$,
    $\texttt{C} \mapsto \texttt{01\$}$,
    $\texttt{G} \mapsto \texttt{10\$}$, and
    $\texttt{T} \mapsto \texttt{11\$}$.
    Here, $r = 3$ and $k = 6$.
    \underline{Left}: $T_6$, the $r$-spaced sparse suffix tree of
    $\Candidates{6} = \{ \texttt{01\$00\$}, \texttt{01\$10\$}, \texttt{10\$01\$} \} \equiv \{ \texttt{CG}, \texttt{GC}, \texttt{CA} \}$ with offset $1$ and round $3$ (i.e., suffixes starting at positions 1 and 4).
    \underline{Right}: the traversal on $s \circ T_6$ for each $s \in \Candidates{6}$.
    Each search begins at the root of the subtree $T_6$.
    At each node $u \in s \circ T_6$  encountered during the traversal, a noisy frequency count 
    $\Tilde{f}=\freqnoisy{\DataSet}{\str{u}}$ is computed (see Section~\ref{sec:noisy_counts} for details). 
    If $\Tilde{f}\geq \freqLowerBound$, then $\str{u}$ is added to $\Candidates{6+t}$, where $6+t=|\str{u}| \in \{9, 12\}$ in this example. 
    Red nodes denote substrings with noisy frequency $\geq$ $\freqLowerBound$, i.e., substrings in 
    $\Candidates{9} = \{ \texttt{01\$10\$01\$}, \texttt{10\$01\$00\$} \} \equiv \{ \texttt{CGC}, \texttt{GCA} \}$ and
    $\Candidates{12} = \{ \texttt{01\$10\$01\$00\$} \} \equiv \{ \texttt{CGCA} \}$.
    Blue nodes (solid dots) and loci within edges (hollow dots) denote points where the search is pruned.
    }
    \label{fig:freq_substring_mining}
\end{figure}

\paragraph{Candidate Searching.}
We search each concatenated tree $s \circ T_k$, $\forall s \in \Candidates{k}$, where $s$ is viewed as a single-path tree whose last node, representing the final character of $s$, is connected to the root of $T_k$.
A search on $s \circ T_k$ begins at node $r$, where $\str{r} = s$ (the root of the subtree $T_k$). 
For each node $u \in s \circ T_k$ at depth $k+t$, we compute a noisy frequency count $\Tilde{f}=\freqnoisy{\DataSet}{\str{u}}$ for the substring represented by the node.
In Section~\ref{sec:noisy_counts}, we show how noisy counts are computed.
If the noisy count $\Tilde{f}$ exceeds threshold $\tau$, we add $\str{u}$ to $\Candidates{k+t}$.
Otherwise, if the noisy count does not exceed the threshold, we prune the search, by returning to node $v=\parent{u}$.
The threshold $\tau$ is chosen so that if $\str{u}$ is $\freqUpperBound$-frequent, then $\Tilde{f} \ge \tau$ with high probability;  
conversely, if $\str{u}$ is $\freqLowerBound$-infrequent, then $\Tilde{f} \ge \tau$ occurs with negligible probability.

Efficiency is provided by the pruning step.
Intuitively, for each $t \in [k]$, only $\freqLowerBound$-frequent strings of length $k + t$ can be added to $\Candidates{k + t}$, and there are at most $O(n \length / \freqLowerBound ) \subseteq O(n)$ such (character-aligned) strings, for $\freqLowerBound \geq \length $.
Therefore, with high probability, across all trees $\{s \circ T_k \mid s \in \Candidates{k}\}$, we only traverse ${O}(n)$ paths, each with length less than or equal to $k$.
In addition, we only construct $T_k$ once and reuse it for subsequent searches during the phase.

This completes the inductive step.
The process continues, eventually completing all phases $i\in \{1, \ldots, \lceil \log \length \rceil\}$.
An example of a phase for $i=2$ is given in Figure~\ref{fig:freq_substring_mining} and pseudocode for the full procedure is available in Algorithm~\ref{alg:DPFS}.

The output of the algorithm is the set of all \emph{well-defined} substrings 
\[
    \bigcup_{j \leq \lengthprime: j \equiv 0 \pmod r }\Candidates{j}.
\]
We can construct this set by building the trie of $\bigcup_{j \leq \lengthprime}\Candidates{j}$ in parallel to candidate generation and testing.
That is, when a new substring is mined, we add the corresponding bit to the output trie.
Leaf nodes that represent undefined substrings can be pruned in a single depth-first search at the conclusion of the algorithm.

\begin{remark}
    Lemma~\ref{lem:frequent_sub} admits a stronger---but also more costly---variant: one may replace condition (ii) by the apparently stronger condition
    \(
        \text{(iii)}\ \bar{s}[t+1\ldots k+t]\in\Candidates{k}.
    \)
    This is the observation exploited by \citet{bernardini2025differentially}.  
    However, enforcing (i) and (iii) simultaneously requires identifying, for different prefix choices, candidate extensions whose length-$t$ suffix itself belongs to $\Candidates{k}$; doing so leads to a blow-up in the number of pairwise checks and dramatically increases running time.
    
    By contrast, Lemma~\ref{lem:frequent_sub} permits an elegant construction of the sparse tree $T_k$ that can be reused for every $s\in\Candidates{k}$ when searching for extensions.  
    This reuse substantially reduces the running time and is a key distinction between our approach and that of \citet{bernardini2025differentially}.
\end{remark}

\subsection{Computing Noisy Frequency Counts}
\label{sec:noisy_counts}

\begin{algorithm}[t]
  \SetAlgoLined
  \DontPrintSemicolon
  \SetKwProg{myproc}{define}{}{}
  \SetKwInOut{require}{require}
  \SetKwInOut{define}{let}
  \require{strings stored and processed using the alphabet $\{\mathsf{0,1,}\$\}$}
    \myproc{$\algname(\DataSet, \eps, \beta, \freqLowerBound)$}
    {
        $r = \lceil\log |\Sigma|\rceil+1; \lengthprime = \length r$; \;
        $\varepsilon_0 =  \frac{ \eps/ \log \length }{4\length\log( n \lengthprime)}; $ \label{line:eps_set}\;
        $\tau \gets 4 \cdot  \varepsilon_0^{-1}\cdot\log (\lengthprime)\ln(n \lengthprime /\beta) + \freqLowerBound$ \;
        \For{$\gamma \in \{\texttt{0},\texttt{1}\}^{r-1}\texttt{\$}$\label{line:init_c1}}
        {
            \If{$\freq{\DataSet}{\gamma} + \laplacenoise(2 \length\eps^{-1} \log \length) > \tau$}
            {
                $\Candidates{r} \gets \Candidates{r}\cup \gamma$\;
            }
        }
        \For{$i \in \{1,\ldots, \lceil\log \length \rceil -1\}$}
        {
            $k \gets r\cdot 2^{i-1}$\; 
            $T_k \gets$ $r$-spaced sparse suffix tree of $\Candidates{k}$ with  round $r$ \; \label{line:tree-init}
            \If{$|\Candidates{k }| >  \frac{n}{\log \lengthprime }$}
            {
               terminate\;
            }
            \For{$s \in \Candidates{k}$}
            {
                Perform depth-first search on $s\circ T_k$ starting from $\trienode{s}$\;
                \ForEach{node $u$ encountered at depth $k+t$}{\label{line:depth}
                    $\Tilde{f} \gets \mathsf{NoisyCount}(s \circ T_k, u, \varepsilon_0)$\;
                    \If{$\Tilde{f} > \tau \land t < k+1$}{
                        $\Candidates{k+t} \gets \Candidates{k+t} \cup \{ \str{u} \}$\;
                        continue depth-first search into children of $u$\;
                    }
                    \Else{
                        prune search at node $u$\;
                    }
                }
            }
        }
       \KwRet $ \bigcup_{j \leq \lengthprime: j \equiv 0 \pmod r }\Candidates{j}$\;
    }
  \caption{Differentially Private Frequent Substring Mining }
  \label{alg:DPFS}
\end{algorithm}

The exact frequency oracle for substrings in $\DataSet$ is constructed with the $r$-spaced sparse suffix tree on the binary encoded string $E(\Document{1}) \#_1^r E(\Document{2}) \#_2^r\ldots E(\Document{n}) \#_1^r$.
Searching for $\freqUpperBound$-frequent items based on true frequencies is not permitted.
However, simply adding Laplace noise to each exact count leads to non-optimal additive error.
For example, at a given search depth $k+t$, by Lemma~\ref{lem:sens_bound_aligned}, the sensitivity of the exact frequency counts is $\mathcal{O}(\length)$.
By Lemma~\ref{lem:lap_mech}, we can assess these candidates while maintaining $\eps^*$-differential privacy and $\mathcal{O}(\length/\eps^*)$ additive error.
However, by Lemma~\ref{lem:composition}, even on a pruned search space, to achieve $\eps$-differential privacy across all phases we require $\eps^* = \eps/\length$ and obtain additive error $\mathcal{O}((\length)^2/\eps)$.

Instead, we adopt a method to estimate substring frequencies, first presented by~\citet{bernardini2025differentially},
which employs a clever use of the Binary Tree mechanism.
We now outline this method, and how the estimates can be computed on-the-fly during each traversal.

Before traversing the candidate trees 
\[
    \mathcal{T}_k = \{s \circ T_k \mid s \in \Candidates{k}\},
\]
we apply a heavy-light decomposition on the base tree $T_k$.
This decomposition depends solely on the topology of $T_k$.
Each heavy edge is chosen to point to the child whose subtree contains the largest number of nodes, with ties broken arbitrarily.
The same decomposition is \emph{reused} across each traversal in $\mathcal{T}_k$ and induces a corresponding decomposition on each candidate tree $s\circ T_k$, partitioning it into vertex-disjoint heavy paths.
By Lemma~\ref{lem:hl-decomposition}, any root-to-leaf path intersects at most $O(\log |T_k|) = O(\log(n\lengthprime))$ heavy paths.

Consider an arbitrary node $u^* \in s \circ T_k$.
Let the heavy path containing $u^*$ be 
\[
 P_{u^*}=(u_0,u_1,\ldots,u_j=u^*,\ldots,u_d).
\]
The frequency of $\str{u^*}$ can be expressed as 
\[
 \freq{\DataSet}{\str{u^*}}= \freq{\DataSet}{\str{u_0}}+\sum_{z=1}^{j} \bigl(\freq{\DataSet}{\str{u_z}}-\freq{\DataSet}{\str{u_{z-1}}})\bigr).
\]
Thus, to estimate  $\freq{\DataSet}{\str{u^*}}$, it suffices to release the DP prefix sums for
\[
(\freq{\DataSet}{\str{u_0}},\; \freq{\DataSet}{\str{u_1}}-\freq{\DataSet}{\str{u_0}},\; \ldots,\; \freq{\DataSet}{\str{u_j} }-\freq{\DataSet}{\str{u_{j-1}}}).
\]
This is precisely what the Binary Tree mechanism provides: a DP counting device for prefix sums.

Each heavy path in the decomposition maintains its own Binary Tree mechanism, named by the root node of the path.
With this structure, we compute the noisy counts on-the-fly, during each traversal.
When a node $u$ is encountered, if it is the root of a heavy path, we initialize a new Binary Tree mechanism with suitably calibrated noise\footnote{Laplace standard deviation $\sigma = \varepsilon_0^{-1} \log \lengthprime$, for $\varepsilon_0 =  \frac{ \eps/ \log \length }{4\length \log( n \lengthprime)}$, suffices.} and increment it with the value $\freq{\DataSet}{\str{u}}$.
Otherwise, if $v$ is the root of the heavy path that contains $u$, we increment the binary-tree mechanism initialized $v$ with the value
\[
     \freq{\DataSet}{\str{{u}}} - \freq{\DataSet}{\str{\parent{u}}} .
\]
This procedure is formalized in Algorithm~\ref{alg:noisy}.
The additive error and privacy provided are detailed in the subsequent section.

\begin{algorithm}[t]
  \SetAlgoLined
  \DontPrintSemicolon
  \SetKwProg{myproc}{define}{}{}
  \SetKwInOut{require}{require}
  \require{Heavy-light decomposition of $T$}
    \myproc{$\mathsf{NoisyCount}(T, u, \varepsilon_0)$}
    {
        $\sigma = \eps^{-1}_0 \cdot \log \lengthprime$\;
        $r \gets$ root of the heavy path that contains $u$\;
        \If{$u = r$}
        {
            Initialize Binary Tree Mechanism $\cB^{\sigma}_r$ with parameter $\sigma$ \tcp*{Lemma~\ref{lem:binary-tree}} 
            $f \gets \freq{\DataSet}{\str{u}}$ \;
        }
        \Else
        {
            $f \gets \freq{\DataSet}{\str{u}}- \freq{\DataSet}{\str{\parent{u}}}$\;
        }
        $\cB^{\sigma}_r.\mathsf{increment}(f)$\;
        \KwRet $\cB^{\sigma}_r.\mathsf{query}$\;
    }
  \caption{Computing noisy count during traversal of candidate tree $T$}
  \label{alg:noisy}
\end{algorithm}

\section{Analysis}
\label{sec:analysis}

Throughout this section we  assume that all input strings are from a binary alphabet.
In cases when the alphabet is not binary, all strings can be converted to binary, increasing the string length to $\lengthprime = \length r$, for $r = \lceil \log |\Sigma| \rceil+1$.

\subsection{Correctness}

\begin{lemma}\label{lem:correctness}
    For any parameters $\beta\in(0,1)$ and $\eps \in \R_{\geq 0}$, set 
    \[
    \eps_0 = \frac{ \eps/ \log \length }{4\length \log( n \lengthprime)}
    \quad \text{and} \quad    
    \tau^* = { \eps_0^{-1} \cdot  \paren{ \log \lengthprime } \cdot \ln(n\lengthprime /\beta)}
    \]
    Then, in Algorithm~\ref{alg:DPFS}, for thresholds
    \begin{align*}
        \freqUpperBound &= 9 \cdot \tau^*  
        \quad \text{and} \quad
        \freqLowerBound \geq \length {\log \lengthprime}
    \end{align*}
    the set $\mathcal{C} = \bigcup_{k=r}^{r \length } \Candidates{k}$ satisfies the $(\freqUpperBound, \freqLowerBound)$-Inclusion-Exclusion Criterion with probability at least $1-\beta$.
\end{lemma}
\begin{proof}
    Let $\beta_0 = \beta/\log \length$.
    The proof proceeds by induction.
    The test threshold in Algorithm~\ref{alg:DPFS} is set at
    \[
        \tau = 4 \cdot \eps_0^{-1} \paren{ \log \lengthprime } \cdot  \ln(n \lengthprime/\beta) + \freqLowerBound = 4\cdot \tau^* + \freqLowerBound.
    \]
    In phase $i=1$, frequent substrings of length $r$ are discovered.
    For each substring $\gamma \in \{\texttt{0},\texttt{1}\}^{r-1}\texttt{\$}$, $\laplacenoise( 2\eps^{-1} \length\log \length)$ noise is added to the true frequency.
    Therefore, by Lemma~\ref{lem:laplace-tail}, with probability $1-\beta_0$, $\forall \gamma \in \in \{\texttt{0},\texttt{1}\}^{r-1}\texttt{\$}$ ,
    \[
        |\freqnoisy{\DataSet}{\gamma}-\freq{\DataSet}{\gamma}| < 2 \eps^{-1}\cdot \length \paren{ \log \length} \ln(2/\beta_0) < \eps_0^{-1}\ln(2 /\beta).
    \]
    This implies that (i) the minimum frequency of any length-$r$ substring assigned to $\Candidates{r}$ is
    \[
        \tau - \eps_0^{-1}\ln(2/\beta) > (4-1) \cdot \tau^* + \freqLowerBound > \freqLowerBound,
    \]
    and (ii) any substring with frequency greater than $\freqUpperBound$ has noisy estimate at least
    \[
        \freqUpperBound -   \eps_0^{-1}\ln(2/\beta) > (9 -1) \cdot \tau^* > \tau,
    \]
    Thus, the $(\freqUpperBound, \freqLowerBound)$-Inclusion-Exclusion Criterion applies to $\Candidates{r}$ for all substrings of length $r$.
    By construction, all substrings in $\Candidates{r}$ are character-aligned.
    This concludes the base case.

    For the inductive case, for any $j<k+t$, where $k = r \cdot 2^i$ and $t\in \{1,\ldots,k\}$, assume that the $(\freqUpperBound, \freqLowerBound)$-Inclusion-Exclusion Criterion applies to $\Candidates{j}$ for character-aligned substrings of length $j$.
    By Lemma~\ref{lem:frequent_sub}, any substring $\bar{s}$ of length $k+t$ with frequency at least $\freqUpperBound$ corresponds to a node $\bar{v}\in s\circ T_k$, for some $s \in \Candidates{k}$, such that $\str{\bar{v}} = \bar{s}$. 
    Moreover, as any prefix $p$ of $\bar{s}$ 
    has frequency $\freq{\DataSet}{p}\geq \freq{\DataSet}{\bar{s}}$, then, by the induction hypothesis, $s^*\in \Candidates{k+t-1}$, where ${s^*} = \str{\parent{\bar{v}}}$.
    Thus, $\bar{s}$ is observed as a candidate during the execution of Algorithm~\ref{alg:DPFS}.

    The noisy estimates of all candidates at depth $k+t$ are generated using the Binary Tree Mechanism, with Laplace standard deviation $\sigma =  \eps_0^{-1}  \log \lengthprime$, on a prefix sum of length at most $\lengthprime$.
    In addition, by the inductive hypothesis, as there are at most $n \length$ character-aligned substrings of length $k+t$ in $\DataSet$, it follows that $|\Candidates{k+t-1}| < n\length/\freqLowerBound< n$.    
    The number of candidate substrings of length $k+t$ is at most $2|\Candidates{k+t-1}|$, as each node that corresponds to a string in $\Candidates{k+t-1}$ has at most two children.
    Let $\mathcal{F} = \{ s \in \children{u} \mid \str{u} \in \Candidates{k+t-1} \}$.
    By Lemma~\ref{lem:binary-tree} and the union bound, $\forall s \in \mathcal{F}$, with probability $1-\beta_0$,
    \begin{align*}
        |\freqnoisy{\DataSet}{s} - \freq{\DataSet}{s}| 
        \leq \sigma \cdot \ln \frac{|\mathcal{F}|}{\beta_0} = \sigma \cdot \ln \frac{2| \Candidates{k+t-1}|}{\beta_0}  
        &\leq 4\cdot \eps_0^{-1}  \paren{ \log \lengthprime} \cdot   \ln(n \lengthprime/\beta) 
        = 4\tau^*
    \end{align*}
    This implies that (i) the minimum frequency of any string assigned to $\Candidates{k+t}$ is at least
    \begin{align*}
        \tau -  4 \cdot \tau^*
        &> (4-4) \cdot \tau^* 
        + \freqLowerBound
        = \freqLowerBound
    \end{align*}
    and (ii) any character with frequency greater than $\freqUpperBound$ has noisy estimate at least
    \begin{align*}
        \freqUpperBound -   4\cdot \tau^* 
        &= (9 -4) \tau^* \geq 4 \cdot \tau^* + \freqLowerBound
        = \tau.
    \end{align*}
    Thus, the $(\freqUpperBound, \freqLowerBound)$-Inclusion-Exclusion Criterion applies to $\Candidates{k+t}$ for all substrings of length $k+t$.
    This concludes the inductive case and the proof.
\end{proof}

As there are at most $n\length$ character-aligned substrings of length $k$, the following result is immediate.
\begin{corollary}
    \label{cor:bound_C_k}
    For parameters $ \beta \in(0,1)$ and  $\freqLowerBound= \length \cdot { \log \lengthprime } $,
    for all $k \in \{0, \ldots, \lengthprime \}$, 
    and with probability at least $1-\beta$,
    \[
        |\Candidates{k}| \le \frac{n}{ { \log \lengthprime }  }.
    \]
\end{corollary}

\subsection{Privacy}\label{subsec:privacy-proof}

\begin{lemma}
    For $\eps \in \R_{\ge 0}$, Algorithm~\ref{alg:DPFS} is $\eps$-differentially private.
    \label{lem:dp_proof}
\end{lemma}

\begin{proof}[Proof of Lemma~\ref{lem:dp_proof}]
    We prove that each phase of \cref{alg:DPFS} is $\eps / \log \length$-DP.  
    By applying the composition lemma (Lemma~\ref{lem:composition}) across the $\log \length$ phases, the overall algorithm is $\eps$-DP.

    \subsubsection*{Phase $1$ (Publishing $\Candidates{r}$)}
    The privacy guarantee follows directly from the sensitivity of $(\freq{\DataSet}{\gamma})_{\gamma \in \{\texttt{0},\texttt{1}\}^{r-1}\texttt{\$}}$, which is at most $2\length$ (Lemma~\ref{lem:sens_bound_aligned}), and from the privacy guarantee of the Laplace mechanism (Lemma~\ref{lem:lap_mech}).
    
    \subsubsection*{Phase $i$ (Publishing $\Candidates{k+1},\ldots,\Candidates{2k}$ for $k = r\cdot 2^{i-1}$)}
        Consider the collection $\HeavyPaths{\Candidates{k}}$ of all heavy paths in $s \circ T_k$, for all $s \in \Candidates{k}$.  
        For each heavy path $\vec{u}= (u_0, u_1, \ldots, u_d) \in \HeavyPaths{\Candidates{k}}$, define the frequency-difference vector
        \[
            \difffreq{\DataSet}{\vec{u}}
            = \Bigl(
                \freq{\DataSet}{\str{u_0}}, 
                \freq{\DataSet}{\str{u_1}} - \freq{\DataSet}{\str{u_0}},\;
                \ldots,\;
                \freq{\DataSet}{\str{u_d}} - \freq{\DataSet}{\str{u_{d-1}}}
            \Bigr).
        \]
        Denote $\cB$ the binary tree mechanism (Lemma~\ref{lem:binary-tree}), with parameter $\eps_0$.
        Let $\DataSet'$ be a neighboring dataset differing in exactly one user's string.
        Then, by Lemma~\ref{lem:binary-tree}, 
        \(
            \cB \paren{\difffreq{\DataSet}{\vec{u}}} 
        \)
        and
        \(
            \cB \paren{\difffreq{\DataSet'}{\vec{u}}}
        \)
        are $\paren{\eps_0 \cdot \norm{\difffreq{\DataSet}{\vec{u}} - \difffreq{\DataSet'}{\vec{u}}}_1}$-indistinguishable (Definition~\ref{def:indistinguishability}).
        By the composition lemma (Lemma~\ref{lem:composition}), the collection
        \[
            \set{\difffreq{\DataSet}{\vec{u}} : \vec{u} \in \HeavyPaths{\Candidates{k}}} 
            \quad\text{and}\quad 
            \set{\difffreq{\DataSet'}{\vec{u}} : \vec{u} \in \HeavyPaths{\Candidates{k}}} 
        \]
        are $\paren{\eps_0 \cdot \sum_{\vec{u} \in \HeavyPaths{\Candidates{k}}} \norm{\difffreq{\DataSet}{\vec{u}} - \difffreq{\DataSet'}{\vec{u}}}_1}$-indistinguishable.  
        If we show that
        \begin{equation}
            \label{eq:heavypath-sens-bound}
            \sum_{\vec{u} \in \HeavyPaths{\Candidates{k}}} \norm{\difffreq{\DataSet}{\vec{u}} - \difffreq{\DataSet'}{\vec{u}}}_1
            \le 4\length \log( n \lengthprime ),
        \end{equation}
        then setting $\eps_0 = \frac{\eps / \log \length}{4 \length \cdot \log( n \lengthprime )} $ (Algorithm~\ref{alg:DPFS} Line~\ref{line:eps_set}) ensures phase $i$ is $\eps / \log \length$-DP, as its output can be viewed as a post-processing of 
        $\set{\difffreq{\DataSet}{\vec{u}} : \vec{u} \in \HeavyPaths{\Candidates{k}}}$.
        It remains to prove \cref{eq:heavypath-sens-bound}. For this, we use the following two lemmas, whose proofs are deferred to Appendix~\ref{app:deferred-proofs}.

        \begin{lemma}[Sensitivity of Heavy-Path Frequency Differences]
            \label{lem:heavy_path_sens}
            Let $\vec{u} = (u_0, \ldots, u_d)$ be a heavy path in $\HeavyPaths{\Candidates{k}}$, and $\DataSet'$ be a neighboring dataset differing in exactly one user's string:
            $\DataSet' = \DataSet \setminus \{\Document\} \cup \{\Document'\}$, where
            $\Document, \Document' \in \{\mathsf{0,1}\}^{1:\lengthprime}$.   
            Then
            \[
                \norm{\difffreq{\DataSet}{\vec{u}} - \difffreq{\DataSet'}{\vec{u}}}_1 
                \le 2 \cdot \max\set{\freq{\Document}{\str{u_0}}, \freq{\Document'}{\str{u_0}}}.
            \]
        \end{lemma}

        \begin{lemma}[Affected Heavy-Path Mass]\label{lem:affected-mass}
        Let $k=r\cdot 2^{i-1}$ in phase $i>1$ and let $\HeavyPaths{\Candidates{k}}$ be the family of heavy paths induced on the candidate trees $\{s\circ T_k : s\in C_k\}$ by the heavy-light decomposition of $T_k$. For neighboring datasets $D' = D\setminus\{S\}\cup\{S'\}$ with $S,S'\in\{0,1\}^{1:\ell_{\mathrm{bit}}}$,
        \[
        \sum_{\vec{u}\in\HeavyPaths{\Candidates{k}}} \max\bigl\{ \freq{\Document}{\mathrm{str}(\mathsf{head}(\vec{u}))},\freq{\Document'}{\mathrm{str}(\mathsf{head}(\vec{u}))} \bigr\}
        \le 2\length\log(n\lengthprime),
        \]
        where $\mathsf{head}(\vec{u})$ is the head of $\vec{u}$. 
        \end{lemma}
        Combining Lemmas~\ref{lem:heavy_path_sens} \& \ref{lem:affected-mass}, we observe
        \begin{align*}
            \sum_{\vec{u}\in \HeavyPaths{\Candidates{k}}} \!\!\! \bigl\|  \difffreq{\DataSet}{\vec{u}}- \difffreq{\DataSet'}{\vec{u}}\bigr\|_1
            &\le \!\!\!\!\! \sum_{\vec{u}\in \HeavyPaths{\Candidates{k}}} \!\!\!\!\! 2 \cdot \max\bigl\{ \freq{\Document}{\mathrm{str}(\mathsf{head}(\vec{u}))},\freq{\Document'}{\mathrm{str}(\mathsf{head}(\vec{u}))} \bigr\} \! 
            \le 4\length\log(n\lengthprime).
        \end{align*}
        This completes the proof.
\end{proof}

\subsection{Performance}

\begin{lemma}\label{lem:performance}
Algorithm~\ref{alg:DPFS}  runs in $O(n \lengthprime + |\Sigma|)$ time and $O(n\length + |\Sigma|)$ space.
\end{lemma}

\begin{proof}
    We decompose the total cost into three components:
    \begin{enumerate}[label=(\roman*)]
        \item\label{cost:gst}
        {constructing the $r$-spaced sparse suffix tree over $\DataSet$, with round $r=\lceil\log |\Sigma|\rceil+1$},
        for exact frequency counts, denoted $\suffixtreeof{\DataSet}$, and annotating each branching node $u$ with the number of leaves in the subtree rooted at $u$ (constructed and annotated once before all phases);
        \item\label{cost:candidate-tree} 
        constructing the tree of character-aligned suffixes in $\Candidates{k}$ (constructed once per phase), 
        denoted $T_k$, and computing its heavy-path decomposition;
        \item\label{cost:traversal}
        in each phase $i$, traversing the concatenated trees $s \circ \trieof{k}$ (for $s \in T_k$) while simultaneously maintaining the corresponding loci in $\suffixtreeof{\DataSet}$.
    \end{enumerate}
    
    \subsubsection*{Runtime}
    Let $E: \Sigma \to \{x \$ \mid x \in \{\texttt{0,1}\}^{r-1} \}$ denote the codebook for converting characters to binary strings.
    The codebook can be stored in $O(|\Sigma|)$ memory.
    For component~\ref{cost:gst}, we construct an $r$-spaced sparse suffix tree $\suffixtreeof{\DataSet}$ on the string $T = E(\Document{1}) \#_1^r \ldots E(\Document{n})\#_n^r $.
    We can construct $T$ in $O(n\length)$ time and space by storing each encoded character in $O(1)$ words.
    By Lemma~\ref{lem:r_spaced_SST_multi},
    constructing $\suffixtreeof{\DataSet}$ takes
    \[
        O \bigparen{
            \sum_{\Document\in\DataSet}|\Document|
        }
        = O( n\lengthprime  )
    \]
    time, and produces a tree with
    \begin{align}
        O \bigparen{
            \sum_{\Document\in\DataSet}|\Document| /r
        }
        = O( n\lengthprime /r ) = O(n \length)
        \label{eqn:nodes_SST}
    \end{align}
    nodes.
    Annotating branching nodes takes a single $O(n \length)$ traversal.
    
    For component~\ref{cost:candidate-tree}, fix phase $i$ and let $k=r \cdot 2^{i-1}$. 
    Let $\Candidates{k} = \{s_1, \ldots, s_{|\Candidates{k}}|\}$.
    Similar to the frequency oracle, $T_k$, the tree of character-aligned tree of suffixes in $\Candidates{k}$, can be implemented by constructing the $r$-spaced sparse suffix tree over the string $s_1 \#_1^r s_2 \#_2^r \ldots s_{|\Candidates{k}|} \#_{|\Candidates{k}}^r$.
    Due to the choice of round, this structure supports a depth-first search on character-aligned suffixes provided we prune at the substring delimiters $\#_i$.
    Therefore, the cost of constructing $T_k$, by Lemma~\ref{lem:r_spaced_SST_multi}, is bounded by 
    \[
        O \PAREN{ \sum_{{s} \in \Candidates{k}} |{s}| }
         = O(|\Candidates{k}| k) 
         = O \PAREN{
            \frac{n \cdot k }{ { \log \lengthprime }  }.
         }
    \]
    by Corollary~\ref{cor:bound_C_k}.
    Summing over all phases yields time complexity
    \[
        O \PAREN{
            \frac{n }{  \paren{ \log \lengthprime }  } 
            \sum_{i=0}^{\left\lceil \log \length \right\rceil} r \cdot 2^{i}
        } 
        = O \PAREN{
            \frac{n \lengthprime }{ { \log \lengthprime }  } 
        }.
    \]  
    
    The heavy-path decomposition of $T_k$ is computed with a single depth-first search, with constant work per node.
    Note that $T_k$ is implemented with compressed (non-branching) paths.
    By definition, each (virtual) edge on this path is heavy, as nodes branching off the path are empty subtrees. 
    Therefore, the cost of computing the heavy-path decomposition is proportional to the number of nodes in the $r$-spaced sparse suffix tree over $\Candidates{k}$.
    Summed across all phases, this takes $O(n \length)$ time.

    For component \ref{cost:traversal}, at phase $1$ we do $|\Sigma|$ constant time noisy frequency queries.
    At phase $i>1$ (with $k=r \cdot 2^{i-1}$) we perform a depth-first search over concatenated trees of the form $s \circ T_k$, where $s \in T_k$.
    The search is pruned as soon as either the substring length exceeds $2k$, or the noisy frequency drops below the test threshold $\tau$.
    Traversal proceeds in the $r$-spaced sparse suffix tree of $\Candidates{k}$.
    At each step, we either move to an actual node in the tree or to a virtual node, with a locus inside a compressed branch.
    Significantly, nodes along a compressed branch can be revealed one-at-a-time, without unfolding the whole branch (which could represent a substring of length $O(k)$).
    Therefore, traversing a branch in $T_k$ takes constant time per revealed character.
    
    At each visited locus, we compute a noisy frequency count for the threshold test.
    Noisy counts are produced online using the Binary Tree mechanism on heavy paths (Algorithm~\ref{alg:noisy}) and require exact substring frequencies for computing the prefix sums.
    To obtain exact frequencies, the depth-first search in $s \circ T_k$ is mirrored in $\suffixtreeof{\DataSet}$.
    We store for each $s \in \Candidates{k}$ a pointer to a node $u_s\in\suffixtreeof{\DataSet}$, where $\str{u_s} = s$, and then resume traversal from $u_s$.
    Each step in $\suffixtreeof{\DataSet}$ takes constant time.
    When traversing a compressed edge in $\suffixtreeof{\DataSet}$ (which could represent a substring of length $O(\lengthprime)$) characters and their frequencies can be revealed incrementally in $O(1)$ time per step.  
    Thus, the cost of producing each exact frequency is $O(1)$.
    
    Each node touched during the traversal of $s\circ T_k$ forms a node in a Binary Tree mechanism.
    With exact frequency information available in constant time, updating and querying each node in each Binary Tree mechanism takes $O(\log d)$ time by Lemma~\ref{lem:binary-tree}, where $d\leq k$ is the length of the corresponding heavy-path.
    As each $\bar{s}\in \Candidates{k+t}$ produces at most two additional candidate nodes and frequency estimates, the total cost of all traversals at phase $i$ is 
    \begin{equation*}
        \label{eq:time-of-the-other}
        \sum_{t=1}^k 2 \cdot \card{\Candidates{k+t}}  \cdot O(\log k) 
        = O \PAREN{
             \frac{n \cdot k \cdot \log k}{ { \log \lengthprime }  } 
        } 
        = O ({n \cdot k})
    \end{equation*}
    time.
    Summing over $\log \length$ rounds the cumulative runtime for all phases $i>1$ is $O(n \lengthprime)$, which is proportional to the time cost for constructing $\suffixtreeof{\DataSet}$. 
    Adding the $O(|\Sigma|)$ cost of phase $1$ gives the stated bound.

    \subsubsection*{Memory}
    For the memory cost, we store $\suffixtreeof{\DataSet}$ throughout algorithm execution and, during phase $i$, for $k=r \cdot 2^i$, we store $T_k$ and a pointer for each string in $\Candidates{k}$.
    By Equation~\ref{eqn:nodes_SST}, the memory required to store $\suffixtreeof{\DataSet}$ is $O(n\length)$.
    The memory required to construct and store $T_k$, by Lemma~\ref{lem:r_spaced_SST_multi} and  Corollary~\ref{cor:bound_C_k}, is bounded by 
    \[
        O \PAREN{ \sum_{{s} \in \Candidates{k}} |{s}| /r}
         = O(|\Candidates{k}| k /r) 
         = O \PAREN{
            \frac{n \cdot k }{ r \cdot { \log \lengthprime }  }
         }
         = O \PAREN{
            \frac{n \cdot \lengthprime }{ r \cdot { \log \lengthprime }  }
         }
         = O \PAREN{
            \frac{n \cdot \length }{  { \log \lengthprime }  }
         }.
    \]
    $T_k$ can be deleted at the end of the phase.
    
    During a traversal of concatenated tree $s \circ T_k$, we can delete a Binary Tree mechanism once we leave the path it belongs to.
    Thus, including the cost of storing the codebook, the space required to execute the algorithm is $O(n \length + |\Sigma|)$. 
\end{proof}

\section{Conclusion}\label{sec:conclusion}

In this work, we studied the problem of differentially private frequent-substring mining and presented a new $\varepsilon$-differentially private algorithm that matches the near-optimal error guarantees of \citet{bernardini2025differentially} while reducing both space complexity from $O(n^2 \length^4)$ to $O(n \length + |\Sigma| )$ and time complexity from $O(n^2 \length^4 + |\Sigma|)$ to $O(n \lengthprime + |\Sigma|)$.
Our refined candidate-generation strategy and frequency-guided pruning rules remove the quadratic blow-ups of prior work, enabling scalable frequent substring mining under differential privacy.
Future work includes extending these techniques to richer pattern-mining tasks and evaluating their empirical performance on large-scale datasets.

\newpage
\bibliographystyle{IEEEtranSN}
\bibliography{reference}


\newpage
\appendix

\section{Deferred Proofs from Section~\ref{subsec:privacy-proof}}\label{app:deferred-proofs}

\begin{proof}[Proof of Lemma~\ref{lem:heavy_path_sens}]
    By definition, for each string $P$, $\freq{\DataSet}{P} = \sum_{\Document^* \in \DataSet} \freq{\Document^*}{P}$.
    Therefore, when $\DataSet' = \DataSet \setminus \{\Document\} \cup \{\Document'\}$,
    \begin{align*}
        \bigl\|
            \difffreq{\DataSet}{\vec{u}}
            - \difffreq{\DataSet'}{\vec{u}}
        \bigr\|_1 
        &= \Bigl|
            \freq{\DataSet}{\str{u_0}} 
            - 
            \freq{\DataSet'}{\str{u_0}} 
        \Bigr|\\
        &+ \! \sum_{j \in [d]}
            \Bigl|
                \bigl(
                    \freq{\DataSet}{\str{u_j}} - \freq{\DataSet}{\str{u_{j-1}}}
                \bigr)
                - 
                \bigl(
                    \freq{\DataSet'}{\str{u_j}} - \freq{\DataSet'}{\str{u_{j-1}}}
                \bigr)
            \Bigr| \\
        &= \Bigl|
            \freq{\Document}{\str{u_0}} 
            - 
            \freq{\Document'}{\str{u_0}} 
        \Bigr|\\
        &+ \! \sum_{j \in [d]}
            \Bigl|
                \freq{\Document}{\str{u_j}} - \freq{\Document}{\str{u_{j-1}}}
                -
                \bigl(
                    \freq{\Document'}{\str{u_j}}
                    - \freq{\Document'}{\str{u_{j-1}}}
                \bigr)
            \Bigr|.
    \end{align*}
    Since each vector 
    $\difffreq{\Document}{\vec{u}}$ and 
    $\difffreq{\Document'}{\vec{u}}$
    has non-positive entries (frequencies only decrease along a heavy path), the triangle inequality gives:
    \begin{align*}
        &\sum_{j \in [d]}
        \Bigl|
            \freq{\Document}{\str{u_j}} - \freq{\Document}{\str{u_{j-1}}}
            -
            \bigl(
                \freq{\Document'}{\str{u_j}}
                - \freq{\Document'}{\str{u_{j-1}}}
            \bigr)
        \Bigr| \\
        &\le
        \sum_{j \in [d]}
            \bigl(\freq{\Document}{\str{u_{j-1}}} - \freq{\Document}{\str{u_j}}\bigr)
        + \sum_{j \in [d]}
            \bigl(\freq{\Document'}{\str{u_{j-1}}} - \freq{\Document'}{\str{u_j}}\bigr) \\
        &\le
        \freq{\Document}{\str{u_0}}
        + \freq{\Document'}{\str{u_0}}.
    \end{align*}
    Therefore, 
    \begin{align*}
        \bigl\|
            \difffreq{\DataSet}{\vec{u}}
            - \difffreq{\DataSet'}{\vec{u}}
        \bigr\|_1 
        &= \Bigl|
            \freq{\DataSet}{\str{u_0}} 
            - 
            \freq{\DataSet'}{\str{u_0}} 
        \Bigr|
        + \freq{\Document}{\str{u_0}}
        + \freq{\Document'}{\str{u_0}} \\
        &\le 2 \cdot \max\set{\freq{\Document}{\str{u_0}}, \freq{\Document'}{\str{u_0}}}.
        \qedhere
    \end{align*}
\end{proof}

\begin{proof}[Proof of Lemma~\ref{lem:affected-mass}]
Fix $\Document$.
Each character-aligned start of a length-$k$ substring in $\Document$ (there are $\le \length$) induces at most one traversal in some $s\circ T_k$. 
By Lemma~\ref{lem:hl-decomposition}, a heavy-light decomposition guarantees that any root-to-leaf traversal crosses $\lceil\log|T_k|\rceil$ heavy paths.
Charge one unit to the head of each crossed path. 
By Corollary~\ref{cor:bound_C_k}, summing over all paths we obtain
\begin{align*}
    \sum_{\vec{u}\in\HeavyPaths{\Candidates{k}}} \freq{\Document}{\str{\mathsf{head}(\vec{u})} }
    &\leq \sum_{s \in \suffix{\Document}} \sum_{\vec{u}\in\HeavyPaths{\Candidates{k}}} \mathds{1}[\str{\mathsf{head}(\vec{u})} \preceq s] \\
    &\leq \length \lceil \log |T_k | \rceil \\
    &\leq 2 \length \log(n \lengthprime).
    \qedhere
\end{align*}
\end{proof}

\clearpage

\section{Related Work}
\label{sec:related-work}

While a rich body of literature exists on differentially private pattern mining, prior approaches fundamentally struggle with utility degradation due to privacy budget splitting or lack formal asymptotic utility guarantees.

\subsubsection*{\bf Contiguous Substring and $q$-gram Mining}
\citet{DBLP:conf/nips/KimGKY21} extract frequent $q$-grams via a top-down search up to length $q$. 
In the notation of our paper, their algorithm can be interpreted as a top-down search for frequent substrings up to length $q$. 
After identifying a candidate set $\Candidates{k}$ of potentially frequent substrings of length $k$ for some $k < q$, they construct $\Candidates{k+1}$ by exploring and pruning substrings in
\(
    \bigl( \Candidates{k} \times \Candidates{1} \bigr)
    \;\cap\;
    \bigl( \Candidates{1} \times \Candidates{k} \bigr).
\)
This requires splitting the privacy budget across $q$ search phases, which in turn degrades the utility guarantees.

Similarly, \citet{DBLP:conf/ccs/ChenAC12} studied the problem of finding frequent $q$-grams for $q=1,\dots,q_{\max}$.  
They adopt a top-down search on the tree whose domain is all possible $q_{\max}$-grams and split the privacy budget across tree heights, using heuristics to adaptively allocate budget among levels.  
Their paper proposes practical allocation strategies, but it does not provide a formal utility analysis of the resulting method.

\citet{DBLP:conf/securecomm/KhatriDH19} studied the private substring-counting problem in genomic data. 
They constructed a reverse tree over all suffixes of the input strings, where each edge is labeled by a symbol from the alphabet, each root-to-node path corresponds to the reverse of some suffix appearing in the dataset, and each node stores the number of occurrences of that suffix. 
They then applied the Laplace mechanism to privatize the node counts in this reverse tree.
Using the privatized counts, they subsequently built a suffix tree in a postprocessing step. 
However, the utility guarantees for answering substring-counting queries on the resulting suffix tree were not explicitly analyzed.

\subsubsection*{\bf Prefix and Trajectory Mining.}
\citet{DBLP:conf/kdd/ChenFDS12} studied frequent sequential pattern mining of public transit data in the Montreal area, representing each passenger's trip as a time-ordered sequence of stations. 
Their method constructs a prefix tree over all transit sequences in a differentially private manner, capturing subsequences that appear frequently as prefixes in the dataset. However, frequent subsequences that do not align with prefixes may not be accurately captured by this approach.

\citet{DBLP:conf/sigmod/ZhangXX16} presented an adaptive top-down approach for tree construction.
Their algorithm was originally designed for spatial data and for quadtree construction.
We adapt their framework to the string setting and briefly summarize the algorithm here.
When the input dataset consists of strings (where two datasets are neighboring if they differ in exactly one string), their algorithm constructs a prefix tree over the strings by retaining only nodes whose frequencies exceed a chosen threshold.
A key insight is that, when the goal is to privately publish only the topology of the prefix tree, their method avoids splitting the privacy budget across the tree height.
However, when prefix-node counts must also be released, their method does not include a formal utility analysis. 

\citet{DBLP:conf/cikm/BonomiX13} studied the private top-$k$ frequent pattern mining problem, aiming to output both frequent prefixes and frequent substrings. Their method adopts a two-phase approach: in the first phase, they construct a differentially private prefix tree to identify frequent prefixes, which then serve as candidates for frequent substrings; in the second phase, they apply fingerprinting techniques to transform and refine these candidates into the final outputs. The first phase requires splitting the privacy budget across different tree levels, and they introduce heuristics to improve the allocation. However, the paper does not provide a formal utility analysis of their method.

\subsubsection*{\bf Non-Contiguous Sequential Patterns.}
\citet{DBLP:journals/tkde/XuCSX016} studied the problem of privately identifying frequent subsequences over a dataset of strings.  
In contrast to substrings, subsequences do not need to be consecutive.  
Their algorithm also follows a top--down search strategy, exploring candidate frequent subsequences in increasing order of length.  
To reduce the search space, they use a subsampled subset of the dataset at each step. However, the paper does not provide an asymptotic utility analysis.

\citet{app12042131} presented a graph-based differentially private method for publishing frequent sequential patterns.
Given a set of frequent patterns obtained in a preprocessing step, they construct a bipartite user–pattern graph and apply an edge-differentially-private noise-graph mechanism, 
then recompute pattern supports from the protected graph to release the patterns with differential privacy guarantees, all without further access to users’ sequences.
They evaluate utility using several empirical metrics and report improvements over prior work; 
however, as with the aforementioned methods, they do not provide a formal asymptotic utility guarantee.

\hfill

\end{document}